\documentclass[submission,copyright,creativecommons]{eptcs}
\providecommand{\event}{CREST 2016} 
\usepackage{breakurl}             

\usepackage{amsmath,amsthm,amssymb,xspace,mathpartir}
\usepackage{graphicx}
\usepackage[all]{xy}

\usepackage{tikz}
\usepackage{tikz-cd}
\usetikzlibrary{matrix,arrows}
\usepackage{DotArrow}
\usetikzlibrary{shapes,arrows,automata}
\usetikzlibrary{calc}
\usetikzlibrary{decorations.pathreplacing}

\newtheorem{theorem}{Theorem}

\newtheorem{corollary}{Corollary}
\newtheorem{lemma}{Lemma}
\newtheorem{definition}{Definition}
\newtheorem{example}{Example}
\newtheorem{remark}{Remark}

\newcommand{\act}{\ensuremath A}

\newcommand{\step}[1]{\xrightarrow{#1}}

\newcommand{\fseq}[1]{{#1}^{*}}
\newcommand{\decS}{\mathcal{D}}

\makeatletter
\newcommand{\xleftrightarrow}[2][]{\ext@arrow 3359\leftrightarrowfill@{#1}{#2}}
\makeatother

\newcommand{\xdasharrow}[2][->]{
\tikz[baseline=-\the\dimexpr\fontdimen22\textfont2\relax]{
\node[anchor=south,font=\scriptsize, inner ysep=1.5pt,outer xsep=2.2pt](x){#2};
\draw[shorten <=3.4pt,shorten >=3.4pt,dashed,#1](x.south west)--(x.south east);
}
}

\makeatletter
\def\xtwoheadrightarrowfill@{%
 \arrowfill@\relbar\relbar{\rightarrow\mkern-15mu\rightarrow}}
\newcommand*\xtwoheadrightarrow[2][]{%
 \ext@arrow 0099\xtwoheadrightarrowfill@{#1}{#2}}
\newcommand{\steps}[1]{\xtwoheadrightarrow{#1}{}}

\newcommand{\eseq}{\varepsilon}

\newcommand{\prc}{\ensuremath{\mathbb S}}

\newcommand{\dia}[1]{\langle #1 \rangle}
\newcommand{\bbox}[1]{[ #1 ]}

\newcommand{\Implies}{\mathrel{\Rightarrow}}
\newcommand{\Par}{\mathrel{||}}

\newenvironment{todo}{\bigskip\hrule\medskip\noindent}{\medskip\hrule\bigskip}

\title{(De-)Composing Causality in Labeled Transition Systems}

\author{Georgiana Caltais
\institute{Department for Computer and Information Science
\\University of Konstanz, Germany}
\email{georgiana.caltais@uni-konstanz.de}
\and
Stefan Leue
\institute{Department for Computer and Information Science
\\University of Konstanz, Germany}
\email{stefan.leue@uni-konstanz.de}
\and
Mohammad Reza Mousavi
\institute{Centre for Research on Embedded Systems\\ Halmstad University, Sweden}
\email{m.r.mousavi@hh.se}
}
\def\titlerunning{(De-)Composing Causality}
\def\authorrunning{G.\ Caltais, S.\ Leue \& M.R. Mousavi}
\begin{document}
\maketitle

\begin{abstract}
In this paper we introduce a notion of counterfactual causality 
in the Halpern and Pearl sense 
that is compositional with respect to the interleaving of transition systems.
The formal framework for reasoning on what caused the violation of a safety property is established in the context of labeled transition systems and Hennessy Milner logic.
The compositionality results are devised for non-communicating systems.
\end{abstract}

\section{Introduction}

Determining and computing causalities is a frequently addressed issue in 
the philosophy of science and engineering, for instance when causally relating system
faults to system failures. 
A notion of causality that is frequently used in relation to technical systems 
relies on counterfactual reasoning. 
Lewis~\cite{Lew73a} formulates the counterfactual argument, which defines when 
an event is considered a cause
for some effect, in the following way: a) whenever the event presumed to be a cause occurs, the effect occurs as well, and
b) when the presumed cause does not occur, the effect will not occur either (counterfactual argument). 
Counterfactual reasoning hence requires the consideration of alternative worlds: one world, corresponding
to one program or system execution in software and systems analysis, where both the cause and the effect 
occur, and another world in which neither the cause nor the effect occur. Cause and effect are 
assumed to be temporally ordered.

In their seminal paper~\cite{halpern2005causes}, Halpern and Pearl argue that the simple Lewis-style counterfactual argument
cannot explain causalities if the causes correspond to complex logical structures of multiple events. Halpern
and Pearl define a notion of complex logical events based on boolean equation systems and propose a number 
of conditions, called actual cause (AC) conditions, under which an event can be considered causal for an effect.
The AC conditions encompass a couterfactual argument.

The Halpern and Pearl model of actual causation has been related in various forms to computing systems.
Most relevant for our work is the work on causality checking~\cite{LeiLeu13d,DBLP:conf/vmcai/Leitner-FischerL13} 
which interprets the Halpern and Pearl event model and notion of actual causation 
in the context of the transition system and trace model for concurrent system computations. In addition to 
the Halpern and Pearl model, in causality checking the order of events as well as the non-occurrence of events
can be causal. An implementation of causality checking using explicit-state model checking~\cite{DBLP:conf/spin/Leitner-FischerL14} 
as well as 
SAT-based bounded model checking~\cite{DBLP:conf/spin/BeerHKLL15} have been provided. The causality checking approach
has been applied to various case studies in the area of analyzing critical systems for safety violations. 
In this setting, an ordered sequence of events is computed as being the actual cause of a safety 
property violation. In safety engineering the safety property violation is usually referred to as 
a hazard.
The computed causalities will be displayed as fault trees complemented by temporal logic formulae which specify the 
order in which causal events occur.



The objective of this paper is to consider the notion of counterfactual causality reasoning 
and actual causation in the context of labeled transition systems (LTS's). In our setting
the LTS's represent system models
and Hennessy Milner logic (HML)~\cite{DBLP:conf/icalp/HennessyM80} formulae specify the system
properties for whose violation actual causes are sought.  We also establish first results on computing causalities
in this setting using (de-)compositional verification.

Our notion of causality complies to the characteristics of "actual causation" proposed in~\cite{halpern2005causes} and further adapted to the setting of concurrent systems in~\cite{DBLP:conf/vmcai/Leitner-FischerL13}. Intuitively, an execution within an LTS is causal whenever it leads to a state where a certain effect, or hazard, is enabled. We handle effects such as the violation of a safety property expressed in HML. Moreover, our definition includes a counterfactual test witnessing that a certain LTS execution $L$ is causal for the occurrence of an effect $E$ if and only if, were $L$ not to happen, $E$ would not occur either.
Additionally, our definition exploits what is referred to as the "non-occurrence of events" in~\cite{DBLP:conf/vmcai/Leitner-FischerL13}, and identifies relevant system execution fragments that, whenever performed, change the occurrence of the effect from true to false. 
Then, similarly to the approaches in~\cite{halpern2005causes,DBLP:conf/vmcai/Leitner-FischerL13}, our definition indicates that a setting that does not include the relevant executions discussed above has no influence on the effect
as long as the causal events are present.
Finally, we require causal executions to be minimal.

We establish the compositionality results with respect to the interleaving of LTS's, thus shifting the fault localization issue to the level of smaller interleaved components.
The current approach only handles non-communicating LTS's.
As an immediate extension of our approach, we would like to extend it to communicating LTS's by adopting ideas from \cite{Aceto12,Gossler15} (please see the conclusions section for more details on this extension).

\paragraph{Related work.}
Lewis-style counterfactual arguments have become the basis for a number of fault analysis, failure localization
and software debugging techniques, such as delta debugging~\cite{Zel09}, nearest neighbor queries~\cite{Renieris},
counterexample explanation in model checking~\cite{Groce03,GroceCKS06} and why-because-analysis~\cite{ladkin1998analysing}.

(De-)compositional verification has been studied in various contexts, such as model-checking~\cite{Andersen95,Giannakopoulou05,Xie05} and model-based conformance testing~\cite{Noroozi13,unknownComponent}.
Our approach is based on our earlier work on decompositional verification of modal mu-calculus formulae~\cite{Aceto12}. 
Regarding compositional verification of causality, we are only aware of the line of work by G{\"{o}}{\ss}ler, Le M{\'{e}}tayer, and associates such as~\cite{Gossler10,Gossler14,Gossler15,Gossler15b}. In the remainder, we review \cite{Gossler10} and \cite{Gossler15} as two closely related examples in this line of work. 

In \cite{Gossler10}, the authors define three trace-theoretic notions of causality for safety properties and provide an assume guarantee framework which allows for decomposing the identification of causes.
They also provide decidability results. Their approach substantially differs from ours: firstly, we combine the different aspects of causality (positive causality, counterfactual, non-occurrence of events, and minimality) in one definition while in \cite{Gossler10} a subset of these aspects is considered in three different definitions. Secondly, the approach of \cite{Gossler10} relies on an assume-guarantee style of specifying the properties, with given LTS models for assume and guarantee contracts, while we rely on the alphabet of the system in decomposing the modal property and its cause. 
Our approach is in its early stages of development and the approach of \cite{Gossler10} has been worked out in various directions. For example, \cite{Gossler10} supports interaction models and is equipped with complexity and decidability results. 

In \cite{Gossler15}, a de-compositional approach to a detecting a trace-based notion of causality is proposed. To start with a failed trace of the system, i.e., a counter-example of the property at hand, is consider and subsequently it is analyzed how the alternative possible behaviors of the different components may lead to failed traces. In our approach, however, we do not start from a system-level counter-example: we aim at decomposing the modal formula for the property, so that all counter-examples are generated locally from the component specifications. Our initial results reported in this paper only concern interleaving components for which a very neat decomposition can be obtained, but our long-term vision is that modal decomposition will enable mechanized decomposition of the modal formula for communicating components, following the approach of \cite{Larsen91,Aceto12}.

A trace-based approach to identifying causality for failures of interleaved systems has been recently introduced in~\cite{DBLP:conf/rv/BefroueiWW14}. In short, the authors propose a method for identifying event sequences that frequently occur within failing system executions, thus possibly revealing causes for system failures. One of the main differences with our approach is that in~\cite{DBLP:conf/rv/BefroueiWW14} system events are parameterised by thread identifiers, program and memory locations, while we consider more abstract events ranging over alphabets denoting (atomic) system actions. Nevertheless, the idea of using thread identifiers might be worth exploited in the context of extending our current work to the setting of concurrent, communicating LTS's.



\paragraph{Paper structure.}{In Section~\ref{sec:prelim} we provide a brief reminder of HML, LTS's, and introduce LTS computations.
In Section~\ref{sec:def-cause} we introduce our notion of causality and provide a series of examples motivating and explaining our definition.
In Section~\ref{sec:dec-cause} we discuss the (de-)compositionality results for causality.
In Section~\ref{sec:conclude} we conclude and provide pointers to further developments.
For a more detailed version of this paper, including complete proofs of the compositionality results, we refer to~\cite{crest-tech-rep}.
}

\section{Preliminaries}\label{sec:prelim}

Let $A$ be a possibly infinite set of labels, usually referred to as \emph{alphabet}. Let $(-)^*$ be the Kleene star operator. We use $w, w_0, w_1, \ldots$ to range over words in $A^*$.
We write $\varepsilon$ for the empty word and $wa$ for the word obtained by concatenating $w \in A^*$ and  $a \in A$. We call a \emph{sub-word} of a word $w$ a word $w'$ obtained by deleting $n$ letters ($n \geq 1$) at some not-necessarily-adjacent positions in $w$, written $w' \in sub(w)$. The empty sequence $\varepsilon$ is a sub-word of $w$.

\begin{definition}[Labeled Transition Systems]
A \emph{labeled transition system} (LTS) is a triple $(\prc,s_0, \act,\rightarrow)$, where $\prc$ is the set of states, $s_0 \in \prc$ is the initial state, $\act$ is the action alphabet and $\rightarrow \subseteq \prc\times\act\times \prc$ is the transition relation.
\end{definition}

We write $\steps{}\subseteq \prc\times\fseq\act\times\prc$, to denote the reachability relation, i.e., the smallest relation satisfying: $\frac{ }{p \steps{\eseq} p},\ \text{and}\ \frac{p \steps{w} p' \quad p' \step a p'' }{p\steps {wa} p''}$.  

The set of actions that can be triggered as a first step from $s \in \prc$ is denoted by $init(s)$: $init(s) = \{a \in A \mid \exists s' \in S~:~s \xrightarrow{a} s'\}$.

\begin{definition}[Computations]\label{def:comp-tr}
Let $[ - ]$ be a list constructor. We write $\decS = [w_0, \ldots, w_n]$ for a finite list of words $w_i \in \act^*$, with $0 \leq i \leq n$. A notation of shape $\decS = [w_0,\, w_1, \ldots ]$ refers to an infinite list $\decS$ of words $w_i \in \act^*$, for $i \geq 0$.
We write $[\,]$ to denote the empty list.
Moreover, we write $w:\decS$ as an alternative to a list with $w$ as the first element, and $\decS$ the "remaining" elements; for instance, $w_1 : [w_2, w_3] = [w_1, w_2, w_3] $.
We say that lists $\decS_0,\ldots ,\decS_n$ are \emph{size-compatible} if they are finite lists of the same length, or if they are all infinite lists. For instance, $[\,]$ and $[\,]$ are size-compatible, $[w_0, w_1, w_2]$ and $[w'_0, w'_1, w'_2]$ are size-compatible, $[w_0, w_1, \ldots]$ and $[w'_0, w'_1, \ldots]$  are size-compatible, whereas $[\,]$ and $[w]$ are not size-compatible.

Consider an LTS $T = (\prc, s_0, \act,\rightarrow)$ and $\pi \in (\prc \times \act \times [ \act^*])^* \times \prc$ a sequence
\[(s_0, l_0, \decS_0), \dots\, (s_n, l_n, \decS_n), s_{n+1}\]
over states $s_i \in \prc$, actions $l_i \in A$ and sets of words $\decS_i \subseteq A^*$, for $0 \leq i \leq n$.
Whenever $\decS_0 , \ldots ,\decS_n$ are size-compatible,
we write $traces((l_0, \decS_0) \ldots (l_n, \decS_n))$ or, in short, $traces(\pi)$, to denote the pairwise extensions of $l_0 \ldots l_n$ with words from $\decS_0, \ldots, \decS_n$ as follows:
\[
\begin{array}{rcl}
traces((l_0, [\,]) \ldots (l_n, [\,])) & = & \{l_0 \ldots l_n\}\\
traces((l_0, w_0:\decS_0) \ldots (l_n, w_n:\decS_n)) & = & \{l_0 w_0 \ldots l_n w_n \} \cup\, 
traces((l_0, \decS_0) \ldots (l_n, \decS_n))
\end{array}
\]
For instance,  $traces((a,[w_{a0}, w_{a1}, w_{a2}]),(b, [\varepsilon, \varepsilon, \varepsilon]), (c, [\varepsilon, w_{c1}, \varepsilon])) = \{a w_{a0} b c, \,a w_{a1} b c w_{c1},\,a w_{a2} b c\}$, for $a, b, c \in A$ and $w_{a0}, w_{a1}, w_{a2}, w_{c1} \in A^*$.

We say that $\pi$ is a \emph{computation} of $T$ whenever the following hold:
\begin{itemize}\itemsep0pt
\item $s_0 \xrightarrow{l_0} s_1 \ldots \xrightarrow{l_n} s_{n+1}$,
\item $\decS_0 , \ldots ,\decS_n$ are size-compatible, and
\item for all $w \in traces(\pi)$ there exists $s \in \prc$ such that $s_0 \steps{w} s$.
\end{itemize}
A computation consisting of only one state $s_0$ is called \emph{trivial computation}.
We use $\pi, \mu, \ldots$ to range over computations.

%

The set of \emph{sub-computations} of $\pi = (s_0, l_0, \decS_0), \ldots, (s_n, l_0, \decS_n), s_{n+1}$, denoted by $sub(\pi)$ is the set of all computations 
$\pi' = (s_0, l'_0, \decS'_0), \ldots, (s_m, l'_m, \decS'_m), s'_{m+1}$ 
such that $l'_0 \ldots l'_m \in sub(l_0 \ldots l_n)$.
Note that all elements of $sub(\pi)$ should be computations themselves. 
\end{definition}

For an intuition, size-compatible lists $\decS_0,\ldots, \decS_n$ encode the pairwise extensions of execution traces $l_0 \ldots l_n$ in $T$ that always disable a certain effect.
Given a computation $(s_0, l_0, \decS_0), \ldots, (s_n, l_n, \decS_n), s_{n+1}$ as above,
sequences $w = l_0 w_0 \ldots l_n w_n \in traces((l_0, \decS_0) \ldots (l_n, \decS_n))$ determine executions
$s_0 \steps{w} s$ in $T$, such that the effect does not occur in $s$.
In our framework, occurrence of effects is formalised in terms of satisfiability of formulae in Hennessy Milner logic~\cite{DBLP:conf/icalp/HennessyM80}.

\begin{definition}[Hennessy-Milner logic]\label{def:HML}

The syntax of Hennessy-Milner logic (HML)~\cite{DBLP:conf/icalp/HennessyM80} is given by the following grammar:
\[\phi, \psi:: = \top \ \mid \ \dia {a} \phi \ \mid \ \bbox {a} \phi \ \mid\ \neg \phi \ \mid \ \phi \land \psi \ \mid \ \phi \lor \psi \ \qquad (a\in\act).\]

We define the satisfaction relation $\vDash$ over LTS's and HML formulae as follows. 
The alphabet of a formula $\phi$, denoted by $alphabet(\phi)$ is the set of actions that appear in $\phi$.

Let $T = (\prc,s_0, \act,\rightarrow)$ be an LTS. Let $\phi,\,\phi'$ range over HML formulae. It holds that:\\[0.5ex]
$
\begin{array}{l}
s  \vDash  \top \textnormal{~~for all $s \in \prc$}\\
s \vDash \neg \phi \textnormal{~~whenever $s$ does not satisfy $\phi$; also written as $s \not \vDash \phi$}\\
s \vDash \phi \land \phi' \textnormal{~~if and only if $s \vDash \phi$ and $s \vDash \phi'$}\\
s \vDash \phi \lor \phi' \textnormal{~~if and only if $s \vDash \phi$ or $s \vDash \phi'$}\\
s \vDash \dia {a} \phi \textnormal{~~if and only if $s \xrightarrow{a} s'$ for some $s' \in \prc'$ such that $s' \vDash \phi$} \\
s \vDash \bbox {a} \phi \textnormal{~~if and only if $s' \vDash \phi$ for all $s' \in \prc'$ such that $s \xrightarrow{a} s'$}.
\end{array}
$
\end{definition}


\section{Defining Causality}\label{sec:def-cause}

We further provide a notion of causality for LTS's. 
The effects that we consider are safety properties 
expressed as HML formulae. Examples motivating and explaining each of the items of our definition are given towards the end of this section.

Our notion of causality complies with that of "actual causation" proposed in~\cite{halpern2005causes} and further adapted to the setting of concurrent systems in~\cite{DBLP:conf/vmcai/Leitner-FischerL13}: 
\begin{itemize}
\item
Intuitively, AC1 in Definition~\ref{def:causality2} states that there must be a setting, or an execution within the LTS under consideration, that determines an effect, or a hazardous situation in which a safety property is violated.
\item
AC2(a) identifies a setting in which the effect does not occur. This is the counter-factual part of our definition.
\item
AC2(b) indicates that, as long as the causal events are present, a setting that does not include the relevant executions discussed above has no influence on the effect.
\item
AC2(c) corresponds to the so-called "non-occurrence of events" in~\cite{DBLP:conf/vmcai/Leitner-FischerL13}, and identifies relevant system execution fragments that, whenever performed, change the occurrence of the effect from true to false. Intuitively, the aforementioned execution fragments are causal by their absence: the effect is enabled only within settings in which the fragments are not executed by our LTS.
\item
AC3 corresponds to the minimality condition in both~\cite{halpern2005causes} and~\cite{DBLP:conf/vmcai/Leitner-FischerL13}.
\end{itemize}

The approach in~\cite{DBLP:conf/vmcai/Leitner-FischerL13} also exploits an ordering condition (OC) that identifies whether the order in which certain events are executed is causal with respect to a given effect, or not. Our framework does not explicitly handle such orderings.
Nevertheless, for non-interleaved systems, such orderings are implicitly captured by sequences $l_0 \ldots l_n$ determined by causal computation as in Definition~\ref{def:causality2}. 
Additionally, as also discussed in Remark~\ref{rm:OC}, the compositionality results in Section~\ref{sec:dec-cause} can alleviate the ordering issue for certain kinds of effects in the context of interleaved systems.

\begin{definition}[Causality for LTS's\label{def:causeLTS}]\label{def:causality2}
Consider a transition system $T$ $=$ $(\prc, s_0, \act,\rightarrow)$; 
causal traces for an HML property $\phi$ in $T$
denoted by $\mathit{Causes}(\phi, T)$ 
is the set of all computations 
$\pi $ $=$ $(s_0, l_0, \decS_0),$ $\ldots,$ $(s_n, l_n, \decS_n), s_{n+1}$  $\in$ $(S \times \act \times [\act^*])^* \times S$ such that 
\begin{enumerate}
\item $s_0 \step{l_0} \ldots s_n \step{l_n} s_{n+1}$ $\land$ $s_{n+1} \vDash \phi$ \textbf{(Positive causality, AC1)}, 

\item $\exists { \chi \in \act^*, 
s' \in \prc} : s_0 \steps{\chi} s' \land s' \vDash \neg \phi$ 
\textbf{(Counter-factual,  AC2(a))},
 



\item[3.]  $\forall \chi' = l_0 \chi_0  \ldots l_n \chi_n \in \{l_0 \ldots l_n\} \cup (A^* \setminus traces((l_0, \decS_0) \ldots (l_n, \decS_n))),\, {s' \in \prc}\,:\,s_{{0}} \steps{\chi'} s' \Implies s' \vDash \phi$ \\
\textbf{(Causality of occurrence,  AC2(b))}


\item[4.]  $\forall \chi' \in traces((l_0, \decS_0) \ldots (l_n, \decS_n)) \setminus \{l_0 \ldots l_n\}, \,{s' \in \prc}  : s_{{0}} \steps{\chi'} s' \Implies s' \vDash \neg \phi$\\
\textbf{(Causality of non-occurrence, AC2(c))}


\item[5.] $\forall {\pi' \in sub(\pi)} : \pi' $ does not satisfy items 1.~--~4. above \textbf{(Minimality, AC3)}
 
\end{enumerate}
\end{definition}

\begin{definition}[Causal projection]\label{def:causal-proj}
A causal projection of $T$ $=$ $(\prc, s_0, \act,\rightarrow)$ with respect to an HML property $\phi$, is $T' = (\prc', s_0, A,\rightarrow')$ such that 
$\prc' = \{ s_i \mid 0 \leq i \leq n+1 \land (s_0, l_0, \decS_0), 
\ldots, (s_n, l_n, \decS_n), s_{n+1} \in$ $\mathit{Causes}(\phi, T)\}$ and 
$\rightarrow' = \{(s_i, l_i, s_{i+1}) \mid 
0 \leq i \leq n \land (s_0, l_0, \decS_0), 
\ldots, (s_n, l_n, \decS_n), s_{n+1} \in$ $\mathit{Causes}(\phi, T)\}$.

We write $T \downarrow \phi$ to denote the causal projection of $T$ with respect to $\phi$.  
\end{definition}

Intuitively, a causal projection is an LTS whose executions capture precisely all causal sequences determined by computations as in Definition~\ref{def:causality2}.


Next, we illustrate the different aspects of Definition \ref{def:causeLTS} using the following small ``canonical'' examples. 
The first example below motivates the positive causality condition (item 1 in Definition \ref{def:causeLTS}).

\begin{figure}
\begin{center}
\scalebox{.8}{
\begin{tabular}{p{4.5cm}p{4.5cm}p{4.5cm}p{4.5cm}}
\includegraphics[bb=0 0 20 170]{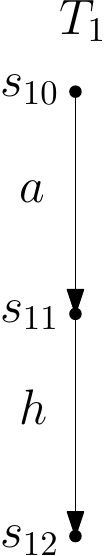}&
\includegraphics[bb=0 0 10 170]{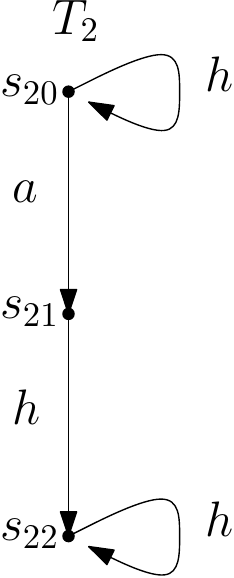}&
\includegraphics[bb=0 0 30 170]{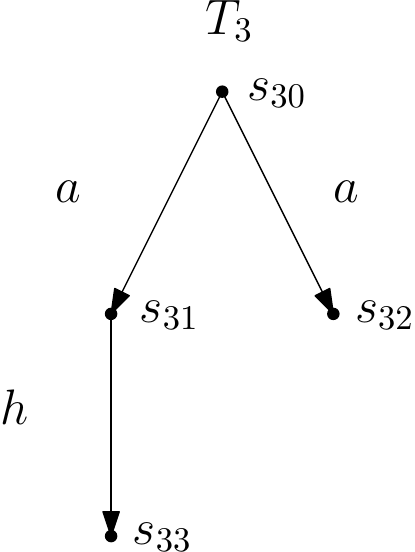}&
\includegraphics[bb=0 0 0 170]{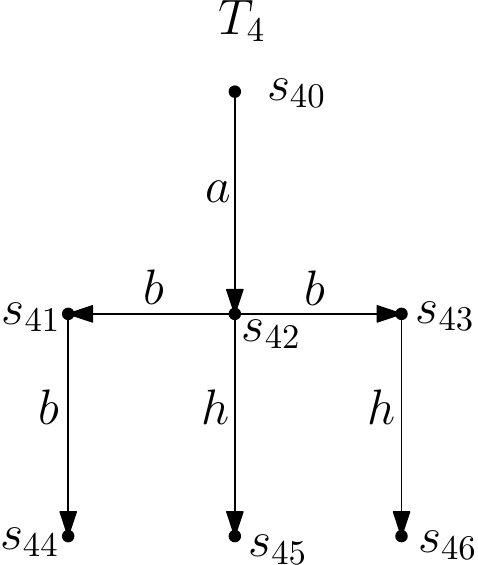} \\
(a) Action $a$ causes hazard $h$. & 
(b) The occurrence of hazard $h$ is factual (trivial). & 
(c) The occurrence of $a$ is not causal for hazard $h$. & 
(d) The non-occurrence of $bb$ is causal for hazard $h$. 
\end{tabular} 
}
\end{center}
\caption{Canonical examples motivating different conditions on causality\label{fig::ah}} 
\end{figure}

\begin{example}[Positive causality]\label{ex:pos-caus}
Consider the formula $\phi = \dia {h} \top$, which states that action $h$ (for hazard) is enabled at the current state and LTS $T_1$ depicted in Figure \ref{fig::ah}.(a). 

The intuition behind the notion of cause suggests that action $a$ should be considered a cause for  $\dia {h} \top$. According to Definition \ref{def:causeLTS}, we have that $(s_{10}, a, [h]), s_{11} \in \mathit{Causes}(\phi, T)$.  The causal projection of $T_1$ for $\phi$ 
is has one transition, namely, $s_{10} \step{a} s_{11}$.  

%

\end{example}

The following example motivates the non-triviality condition (item 2 in Definition \ref{def:causeLTS}).

\begin{example}[Counter-factual]
Consider the LTS $T_2$ depicted in Figure \ref{fig::ah}.(b) and the same formula $\phi = \dia{h} \top$.
Although trace $a$ can lead to a state where $\phi$ holds, the hazard formula holds trivially everywhere else, and hence there is no cause to be identified; we refer to Lemma~\ref{lm:dec-disj-trivial} for a formalisation. 
\end{example}

The next two examples motivate the causality of occurrence and non-occurence, respectively (items 3 and 4 in Definition \ref{def:causeLTS}).

\begin{example}[Causality of Occurrence\label{ex::causalityOccurrence}]
Consider the LTS $T_3$ depicted in Figure \ref{fig::ah}.(c) and the same formula $\phi = \dia{h} \top$.
Trace $a$ can non-deterministically lead to two states, namely $s_{31}$ and $s_{32}$. 
The formula holds only in one of them, 
namely in $s_{31}$. Hence, $a$ cannot be considered a cause for the hazard. 
More precisely, if a trace is causal then its execution, or ``occurrence", always leads to a state where the hazard holds. 

\end{example}
 

%

\begin{example}[Causality of Occurrence and Non-occurrence\label{ex:causalityOccurNonOcur}]

Consider the LTS $T_4$ depicted in Figure \ref{fig::ah}.(d) and the same formula $\phi = \dia{h} \top$.
Trace $a$ leads to state $s_{42}$ where the hazard formula holds.
Trace $ab$ also leads to a hazardous state $s_{43}$; however, performing another $b$, i.e., performing the trace $abb$ from the initial state, removes the hazard. 
Hence, $(s_{40}, a, [\varepsilon]), s_{42}$ is not in the set of causes for $\phi$, because extending $a$ with $bb$, for instance, violates $\phi$ and thereby violating item 3 in Definition 
 \ref{def:causeLTS}.
However,  $(s_{40}, a, [h, bb, bh]), s_{42}$ is a cause, because $a$  leads to a hazard, all possible extensions of $a$ with anything but $h$, $bb$ or $bh$, the only ones being $\varepsilon$ and $b$, also keep the hazard. On the other hand, the extensions of $a$ with $h$, $bb$ or $bh$ remove the hazard. Hence, $h$, $bb$ and $bh$ are the "relevant extension" that enable removing the hazard.
\end{example}

The next example motivates the minimality condition, item 5 in Definition \ref{def:causeLTS}. 
  
\begin{example}[Minimality Condition]\label{ex:min-cond}
Consider again the LTS $T_4$ treated in Example \ref{ex:causalityOccurNonOcur}. 
Computation $(s_{40}, a, [\varepsilon,\varepsilon]),$ $(s_{42}, b, [h,b]), s_{43}$  is not a cause because it is not minimal (violating item 5 in Definition 
\ref{def:causeLTS}). This is because its sub-computation $(s_{40}, a, [h,bb,bh]),$ $s_{42}$ is a cause as illustrated in Example~\ref{ex:causalityOccurNonOcur}. 

Consider the LTS $T_5$ depicted in Figure \ref{fig::ah2}.(a) and the formula $\phi = \dia{h} \top$.
For instance, the computation $(s_{50}, a, [\varepsilon, \varepsilon, \varepsilon \ldots]),$ $(s_{51}, i, [h, ih, iih\ldots]),$ $s_{51}$ is not in $\mathit{Causes}(\phi, T_5)$, because performing an $i$ does not change the state of the system and hence, cannot contribute to the occurrence of the hazard. Computation $(s_{50}, a, [h, ih, iih, \ldots]),$ $s_{51}$,  however, is in $\mathit{Causes}(\phi, T_5)$, because it satisfies all the conditions of the cause, including minimality.

%

Consider the LTS $T_6$ depicted in Figure \ref{fig::ah2}.(b) and the formula $\phi = \dia{h} \top$.
Computation $(s_{60}, a, [\varepsilon]),$ $(s_{63}, b, [h]),$ $s_{65}$ 
is a cause for $\phi$, despite the fact that computation $(s_{60}, a, [h,bh]), s_{61}$ also leads to the hazard.

This is not a violation of minimality, because $(s_{60}, a, [h,bh]), s_{61}$  does not satisfy the so-called ''Causality of occurrence" (AC2(b)) in Definition~\ref{def:causality2}, as also illustrated in Example~\ref{ex::causalityOccurrence}. 
\end{example}

\begin{figure}
\begin{center}
\scalebox{.8}{
\begin{tabular}{p{4.5cm}p{4.5cm}}
\includegraphics[bb=0 0 0 150]{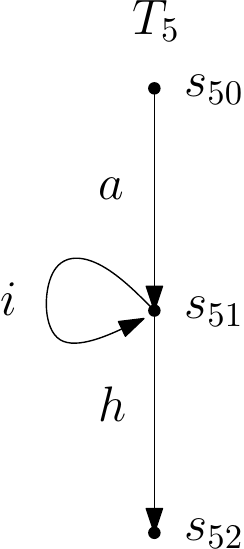} & 
\includegraphics[bb=0 70 0 150]{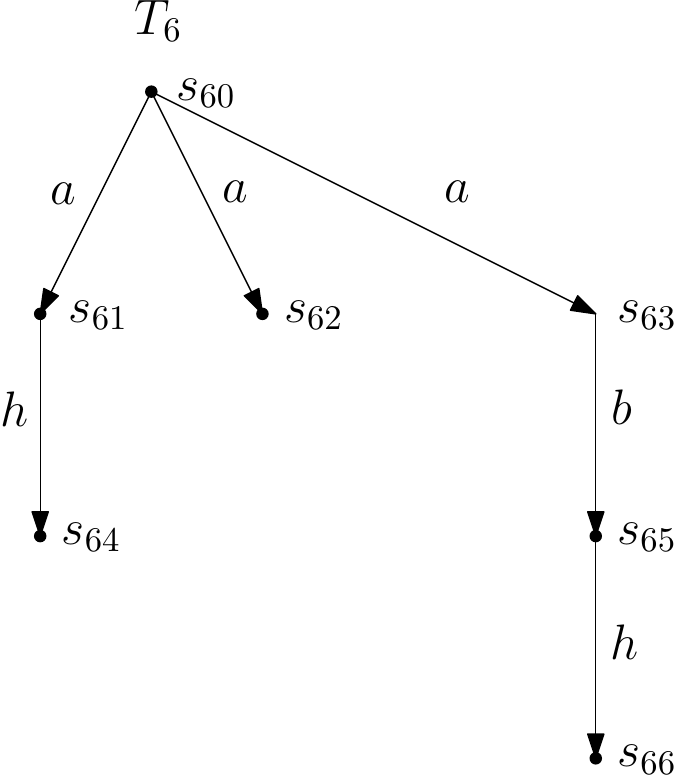} 
  \\
(a) Action $i$ does not contribute to $h$. & 
(a) Trace $ab$ is a cause because trace $a$ is not a cause.
\end{tabular} 
}
\end{center}
\caption{Canonical examples motivating minimality condition\label{fig::ah2}} 
\end{figure}

\section{(De-)composing Causality} \label{sec:dec-cause}

In this section we provide the main results regarding (de-)compositionality of causality.
Theorem~\ref{thm:decomposing-disjunction-iso} states the equivalence between reasoning on causality with respect to disjunctions $\phi \lor \psi$ of HML formulae in the context of interleaved LTS's, and reasoning on causality with respect to $\phi$ or $\psi$ in the corresponding interleaved components.
Orthogonally, Theorem~\ref{thm:decomposing-conjunction-iso} captures the equivalence between reasoning on causality with respect to conjunctions $\phi \land \psi$ of HML formulae in the context of interleaved LTS's, and reasoning on causality with respect to $\phi$ and $\psi$ in the corresponding interleaved components.
Both results are established for non-communicating LTS's executing disjoint sets of actions.

Our formal framework exploits standard notions of interleaving ($\Par$) and non-deterministic ($+$) choice between LTS's~\cite{DBLP:books/sp/Milner80} or, more explicitly, between causal projections as in Definition~\ref{def:causal-proj}.
Consider the LTS $T = (\prc,s_0,\act,\rightarrow)$, 
$a \in A$ and $s, s', p, p' \in \prc$. Then:
\[
\begin{array}{ll}
s \Par p \xrightarrow{a} s' \Par p \textnormal{~~whenever~~} s \xrightarrow{a} s' & \hspace{20pt}
s + p \xrightarrow{a} s' \textnormal{~~whenever~~} s \xrightarrow{a} s'\\
s \Par p \xrightarrow{a} s \Par p' \textnormal{~~whenever~~} p \xrightarrow{a} p' & \hspace{20pt}
s + p \xrightarrow{a} p' \textnormal{~~whenever~~} p \xrightarrow{a} p'.
\end{array}
\]

Consider LTS's $T = (\prc,s_0,\act,\rightarrow)$ and $T' = (\prc',s_0',B,\rightarrow')$.
We abuse the notation and write $T || T'$ in lieu of $s_0 \Par s'_0$, and
$T + T'$ in lieu of $s_0 + s'_0$.

With this intuition in mind, we proceed to discussing our compositionality results. 

Lemma~\ref{lm:dec-disj-trivial} provides a result that shows that reasoning on (de-)composition of causality in the context of formulae that hold in the initial state of a system is trivial.

\begin{lemma}[Immediate Causality]\label{lm:dec-disj-trivial}
Consider the LTS's $T = (\prc,s_0,\act,\rightarrow)$ and the HML property $\phi$.
If $s_0 \vDash \phi$ it holds that  
$
s_0  = \mathit{Causes}(\phi, T) 
$
or $\mathit{Causes}(\phi) = \emptyset$.
\end{lemma}

We call properties $\phi$ as above \emph{immediate effects}.

\subsection{(De-)Composing Disjunction}\label{sec:dec-disj}
In what follows we show that reasoning on causality with respect to disjunctions of HML formulae $\phi \lor \psi$ can be performed in a compositional fashion.

Intuitively, the result in Lemma~\ref{lm:decomposing-disjunction-only-if} states that causality is preserved under disjunction of HML formulae and the interleaving of non-communicating LTS's.
Or, more precisely, given two non-communicating LTS's $T$ and $T'$ and two HML formulae $\phi$ and $\psi$ built over their corresponding alphabets, it holds that a cause $\pi \in \mathit{Causes}(\phi, T)$ determines a cause $\mu \in \mathit{Causes}(\phi \lor \psi, T \Par T')$ within the interleaved LTS's.

\begin{lemma}[]\label{lm:decomposing-disjunction-only-if}
Consider LTS's $T = (\prc,s_0,\act,\rightarrow)$ and $T' = (\prc',s_0',B,\rightarrow')$ such that $A \cap B = \emptyset$.
Assume two HML formulae $\phi$ and $\psi$ over $A$ and $B$, respectively.
Whenever $\phi$ and $\psi$ are not immediate effects, the following holds:
\begin{enumerate}
\item[ ] If $\pi = (s_0 , l_0, \decS_0), \ldots,  (s_n , l_n, \decS_n), s_{n+1} \in \mathit{Causes}(\phi, T )$, then there exists\\
$\mu = (s_0 \Par s'_0, l_0, \overline{\decS}_0), \ldots,  (s_n \Par s'_0, l_n, \overline{\decS}_n), s_{n+1} \Par s'_0 \in \mathit{Causes}(\phi \lor \psi, T \Par T')$.
\end{enumerate}
\end{lemma}
\begin{proof}[Proof Sketch.]
The statement follows by two intermediate results.

We show how to create a computation $\mu$ satisfying conditions AC1--AC2(c) in Definition~\ref{def:causality2} from $\pi$, given the hypothesis that $\pi$ satisfies conditions AC1--AC2(c) as well.
AC1 is satisfied for $\mu$ as a consequence of AC1 being satisfied for $\pi$. AC2(a) trivially holds for $\mu$ as $\phi$ and $\psi$ are not immediate effects. Showing AC2(b) and AC2(c) strongly relies on the shape of $\overline{\decS}_0, \ldots, \overline{\decS}_n$. The lists $\overline{\decS}_i$ are created in three steps.
\begin{enumerate}
\item We begin by simply "copying" the information in each $\decS_i$ into the corresponding $\overline{\decS}_i$.
\item We identify all causal traces $\chi$ obtained by interleaving the causal traces of $\pi$ with the causal traces determined by all computations in $\mathit{Causes}(\psi, T')$. We make the necessary insertions into the lists $\overline{\decS}_i$, so that $\chi$'s are stored as causal traces of computations in $\mathit{Causes}(\phi \lor \psi, T \Par T')$.
\item We compute all the causal traces $\chi$ for $\phi \lor \psi$ that do not allow $s'_0$ to evolve in $T'$, but consist of words in $B$ as well. We make the necessary insertions into the lists $\overline{\decS}_i$, so that $\chi$'s are stored as causal traces of computations in $\mathit{Causes}(\phi \lor \psi, T \Par T')$. This step guarantees that the remaining traces in $(A \cup B)^* \setminus traces((l_0, \overline{\decS}_0)\ldots((l_n, \overline{\decS}_n)))$ are not "harmful" with respect to AC2(b) for $\mu$, as they never lead to $s \Par s' \vDash \neg \phi \land \neg \psi$.
\end{enumerate}
By the above construction, AC2(b) and AC2(c) hold for $\mu$ as well.

AC3 for $\mu$ is proved to hold by reductio ad absurdum. In short, we show that whenever there is $\mu' \in sub(\mu)$, such that $\mu'$ satisfies AC1--AC2(b), there exists $\pi' \in sub(\pi)$, such that $\pi'$ satisfies AC1--AC2(b) as well. This contradicts the hypothesis $\pi \in \mathit{Causes}(\phi, T )$. 

\end{proof}

Intuitively, Lemma~\ref{lm:decomposing-disjunction-if} states that causality with respect to an effect $\phi \lor \psi$ in two interleaved, but non-communicating LTS's, is preserved by at least one of the interleaved components.
Or, more precisely, given two non-communicating LTS's $T$ and $T'$ and two HML formulae $\phi$ and $\psi$ built over their corresponding alphabets, it holds that a cause
$\mu \in \mathit{Causes}(\phi \lor \psi, T \Par T')$ within the interleaved LTS's
determines 
a cause $\pi \in \mathit{Causes}(\phi, T)$ for $\phi$ in $T$, or a cause $\pi'  \in \mathit{Causes}(\psi, T')$ for $\psi$ in $T'$.

\begin{lemma}[]\label{lm:decomposing-disjunction-if}
Consider LTS's $T = (\prc,s_0,\act,\rightarrow)$ and $T' = (\prc',s_0',B,\rightarrow')$ such that $A \cap B = \emptyset$. Assume two HML formulae $\phi$ and $\psi$ over $A$ and $B$, respectively.
Whenever $\phi$ and $\psi$ are not immediate effects, the following holds:
\begin{enumerate}
\item[ ] If $\mu = (s_0 \Par s'_0, l_0, \decS_0), \ldots,  (s_n \Par s'_n, l_n, \decS_n), s_{n+1} \Par s'_{n+1} \in \mathit{Causes}(\phi \lor \psi, T \Par T')$, then there exists\\
$\pi = (s_k , l_k, \overline{\decS_k}), \ldots,  (s_m , l_m, \overline{\decS_m}), s_{n+1} \in \mathit{Causes}(\phi, T )$ or\\
$\pi' = (s'_p , l'_p, \overline{\decS'_p}), \ldots,  (s'_q , l'_q, \overline{\decS'_q}), s'_{n+1} \in \mathit{Causes}(\psi, T' )$.
\end{enumerate}
For all $k \leq i \leq m$: $(s_i, l_i, \overline{\decS}_i)$ corresponds to $(s_i \Par s'_i, l_i, {\decS}_i)$ in $\mu$, whenever $l_i \in A$.
For all $p \leq j \leq q$: $(s'_j, l'_j, \overline{\decS'}_j)$ corresponds to $(s_j \Par s'_j, l'_j, {\decS'}_j)$ in $\mu$, whenever $l'_j \in B$.
Moreover, $l_k \ldots l_m = l_0 \ldots l_n \downarrow A$, $l'_p \ldots l'_q = l_0 \ldots l_n \downarrow B$.
\end{lemma}
\begin{proof}[Proof Sketch.]
The statement follows by two intermediate results.

First, we show that one can build $\pi$ or $\pi'$ as above, such that $\pi$ or $\pi'$ satisfy conditions AC1--AC2(c) in Definition~\ref{def:causality2}, given the hypothesis that $\mu$ satisfies AC1--AC2(c) as well.
The reasoning for proving this intermediate result strongly relies on the shape of the lists $\overline{\decS}_i$ and $\overline{\decS}'_j$ corresponding to $\pi$ and $\pi'$, respectively.
We construct the aforementioned lists in three steps.
\begin{enumerate}
\item We start with empty lists $\overline{\decS}_i$ and $\overline{\decS}'_j$.
\item Then, we "encode" causal sequences $\chi \in traces((l_0, \decS_0)\ldots(l_n, \decS_n)) \setminus \{l_0 \ldots l_n\}$ satisfying AC2(c) by definition, into $traces((l_k, \overline{\decS}_k)\ldots(l_m, \overline{\decS}_m))$ and, respectively, $traces((l'_p, \overline{\decS}'_p)\ldots(l'_q, \overline{\decS}'_q))$, via the projections of $\chi$ on $A$ and, respectively, $B$ that satisfy AC2(c) as well.
\item Eventually, we "prepare" $\pi$ for satisfying AC2(b). 
We identify all sequences $\chi \in A^* \setminus traces((l_k, \overline{\decS}_k)$
$\ldots(l_m, \overline{\decS}_m))$ that always lead to $s \vDash \neg \phi$. 
For each such $\chi$ we make the necessary insertions into the lists $\overline{\decS}_i$, so that $\chi$'s are stored as causal traces of computations in $\mathit{Causes}(\phi, T)$. We repeat the "preparation" process for $\pi'$ as well.
\end{enumerate}

Then, we show that $\pi$ or $\pi'$ satisfy AC1--AC2(c) by reductio ad absurdum.
Without loss of generality, assume that $\pi$ satisfies AC1--AC2(c). Showing that $\pi$ has to satisfy $AC3$ as well follows by proof by contradiction. More explicitly, we show that whenever there exists $\widetilde{\pi} \in sub(\pi)$ satisfying AC1--AC2(c), one can construct $\widetilde{\mu} \in sub(\mu)$ such that $\widetilde{\mu}$ satisfies AC1--AC2(c) as well. This contradicts the hypothesis $\mu \in \mathit{Causes}(\phi \lor \psi, T \Par T')$.

\end{proof}

Corollary~\ref{cor:decomposing-disjunction-if} states that a causal computation $\mu$ with respect to an effect $\phi \lor \psi$ in interleaved, but non-communicating LTS's, determines a causal computation $\pi$ in the interleaved component that triggered the first step in $\mu$.

\begin{corollary}[]\label{cor:decomposing-disjunction-if}
Consider LTS's $T = (\prc,s_0,\act,\rightarrow)$ and $T' = (\prc',s_0',B,\rightarrow')$ such that $A \cap B = \emptyset$.
Assume two HML formulae $\phi$ and $\psi$ over $A$ and $B$, respectively.
Whenever $\phi$ and $\psi$ are not immediate effects, the following holds:
\begin{enumerate}
\item[ ] If $\mu = (s_0 \Par s'_0, l_0, \decS_0), \ldots,  (s_n \Par s'_n, l_n, \decS_n), s_{n+1} \Par s'_{n+1} \in \mathit{Causes}(\phi \lor \psi, T \Par T')$ then
\begin{itemize}
\item if $l_0 \in A$ then exists $\pi = (s_k , l_k, \overline{\decS_k}), \ldots,  (s_m , l_m, \overline{\decS_m}), s_{n+1} \in \mathit{Causes}(\phi, T )$; otherwise
\item if $l_0 \in B$ then exists $ \pi' = (s'_p , l'_p, \overline{\decS'_p}), \ldots,  (s'_q , l'_q, \overline{\decS'_q}), s'_{n+1} \in \mathit{Causes}(\psi, T' )$.
\end{itemize}
\end{enumerate}
For all $k \leq i \leq m$: $(s_i, l_i, \overline{\decS}_i)$ corresponds to $(s_i \Par s'_i, l_i, {\decS}_i)$ in $\mu$, whenever $l_i \in A$.
For all $p \leq j \leq q$: $(s'_j, l'_j, \overline{\decS'}_j)$ corresponds to $(s_j \Par s'_j, l'_j, {\decS'}_j)$ in $\mu$, whenever $l'_j \in B$.
Moreover, $l_k \ldots l_m = l_0 \ldots l_n \downarrow A$, $l'_p \ldots l'_q = l_0 \ldots l_n \downarrow B$.
\end{corollary}
\begin{proof}
The result follows immediately by Lemma~\ref{lm:decomposing-disjunction-if}, Lemma~\ref{lm:decomposing-disjunction-only-if} and the minimality condition AC3 in Definition~\ref{def:causality2}.
\end{proof}

Lemma~\ref{lm:decomposing-disjunction-iso} states that, as a consequence of the minimality condition, causal computations with respect to effects $\phi \lor \psi$ in interleaved, non-communicating LTS's capture executions of only one of the interleaved components.

\begin{lemma}[]\label{lm:decomposing-disjunction-iso}
Consider LTS's $T = (\prc,s_0,\act,\rightarrow)$ and $T' = (\prc',s_0',B,\rightarrow')$ such that $A \cap B = \emptyset$.
Assume two HML formulae $\phi$ and $\psi$ over $A$ and $B$, respectively.
Whenever $\phi$ and $\psi$ are not immediate effects and $\mu \in \mathit{Causes}(\phi \lor \psi, T \Par T')$, then either
\begin{enumerate}
\item[ ] $\mu = (s_k \Par s'_0, l_k, {\decS}_k), \ldots,  (s_m \Par s'_0, l_m, {\decS}_m), s_{n+1} \Par s'_0$, 
or\\
$\mu = (s_0 \Par s'_p, l'_p, {\decS}'_p), \ldots,  (s_0 \Par  s'_q, l'_q, {\decS}'_q), s_0 \Par  s'_{n+1}$
\end{enumerate}
such that, for all $k \leq i \leq m$ and $p \leq j \leq q$: $s_i \in \prc$, $s'_j \in \prc'$, $l_i \in A$, $l'_j \in B$,
$\decS_i \in A^*$ and $\decS'_j \in B^*$.
\end{lemma}
\begin{proof}
Assume $\mu = (s_0 \Par s'_0, l_0, \overline{\decS}_0), \ldots,  (s_n \Par s'_n, l_n, \overline{\decS}_n), s_{n+1} \Par s'_{n+1} \in \mathit{Causes}(\phi \lor \psi, T \Par T')$.
Assume, without loss of generality, that by Lemma~\ref{lm:decomposing-disjunction-if} there exists a computation:
\[ \widetilde{\pi} = (s_k , l_k, \widetilde{\decS}_k), \ldots,  (s_m , l_m, \widetilde{\decS}_m), s_{n+1} \in \mathit{Causes}(\phi, T )\]
such that
for all $k \leq i \leq m$: $(s_i, l_i, \widetilde{\decS}_i)$ corresponds to $(s_i \Par s'_i, l_i, \overline{\decS}_i)$ in $\mu$, whenever $l_i \in A$.
Moreover, $l_k \ldots l_m = l_0 \ldots l_n \downarrow A$.
Then, by Lemma~\ref{lm:decomposing-disjunction-only-if}, it follows that there exists a computation
\[\widehat{\mu} = (s_k \Par s'_0, l_k, \widehat{\decS}_k), \ldots,  (s_m \Par s'_0, l_m, \widehat{\decS}_m), s_{n+1} \Par s'_0 \in \mathit{Causes}(\phi \lor \psi, T \Par T').\]
Additionally, observe that $\widehat{\mu} \in sub(\mu)$. This violates the minimality condition AC3 for $\mu$, unless $\mu = \widehat{\mu}$. This proves our initial statement.
\end{proof}

Theorem~\ref{thm:decomposing-disjunction-iso} is the main result of this section.
Intuitively, it states that reasoning on causality with respect to an effect $\phi \lor \psi$ in the context of non-communicating, interleaved LTS's is equivalent to reasoning on causality for $\phi$ or $\psi$ in the context of the corresponding interleaved components.

\begin{theorem}[(De-)composing Disjunction]\label{thm:decomposing-disjunction-iso}
Consider LTS's $T = (\prc,s_0,\act,\rightarrow)$ and $T' = (\prc',s_0',B,\rightarrow')$ such that $A \cap B = \emptyset$.
Assume two HML formulae $\phi$ and $\psi$ over $A$ and $B$, respectively.
Whenever $\phi$ and $\psi$ are not immediate effects, the following holds: 
\begin{equation}
T \Par T' \downarrow (\phi \lor \psi) \,\,\simeq\,\, 
T \downarrow \phi + T' \downarrow \psi.
\end{equation}
\end{theorem}

\begin{proof}
Let $(\prc_{\Par},s_{0} \Par s'_0,A \cup B,\rightarrow_{\Par}) = (T \Par T') \downarrow (\phi \lor \psi)$ and $(\prc_{+},s_{0} + s'_0,A \cup B,\rightarrow_{+}) = (T \downarrow \phi) + (T' \downarrow \psi)$, respectively.
The result follows immediately by Corollary~\ref{cor:decomposing-disjunction-if}, Lemma~\ref{lm:decomposing-disjunction-iso} and the semantics of the non-deterministic choice operator ($+$), where the isomorphic structure is underlined by:
\[
\begin{array}{cc}
f : \prc_{\Par} \rightarrow \prc_{+} & \hspace{30pt} f^{-1}:\prc_{+} \rightarrow \prc_{\Par}\\[0.5ex]
\begin{array}{rcl}
f(s_0 \Par s'_0) & = & s_0 + s'_0\\
f(p \Par q) & = & \left \{
\begin{array}{lcl}
p & \textnormal{if} & q = s'_0 \land p \not = s_0\\
q & \textnormal{if} & p = s_0 \land q \not = s'_0
\end{array}
\right.
\end{array} &
\begin{array}{rcl}
f^{-1}(s_0 + s'_0) & = & s_0 \Par s'_0\\
f^{-1}(p) & = & \left \{
\begin{array}{lcl}
p \Par s'_0 & \textnormal{if} & p \in \prc \land p \not = s_0\\
s_0 \Par p & \textnormal{if} & p \in \prc' \land p \not = s'_0
\end{array}
\right.
\end{array}
\end{array}
\]
\end{proof}

\begin{example}\label{ex:dec-disj}
For an example, consider two LTS's $T$ and $T'$ with initial states $s_0$ and $p_0$, respectively, depicted as in Figure~\ref{fig:ex-dec}.
Let $\phi = \dia{h} \top$ and $\psi = \dia{h'} \top$ be two HML formulae.
It is straightforward to see that $T \downarrow \phi$ is defined by the dotted transition $s_0 \dotarrow{a} s_1$ in $T$, whereas $T' \downarrow \psi$ is $p_0 \dotarrow{d} p_1 \dotarrow{e} p_2$.
The interleaving of $T$ and $T'$ is the LTS originating in $s_0 \Par p_0$ in Figure~\ref{fig:ex-dec}.
At a closer look, one can see that $T \Par T' \downarrow (\phi \lor \psi)$ is the transition system defined by the dotted transitions  $s_0\Par p_0 \dotarrow{a} s_1 \Par p_0$ and $s_0 \Par p_0 \dotarrow{d} s_0 \Par p_1 \dotarrow{e} s_0 \Par p_2$, which is obviously isomorphic with $T \downarrow \phi + T' \downarrow \psi$.


\begin{figure}
\scalebox{.5}{
\begin{tabular}{p{4.5cm}}
\includegraphics[bb=0 0 0 550]{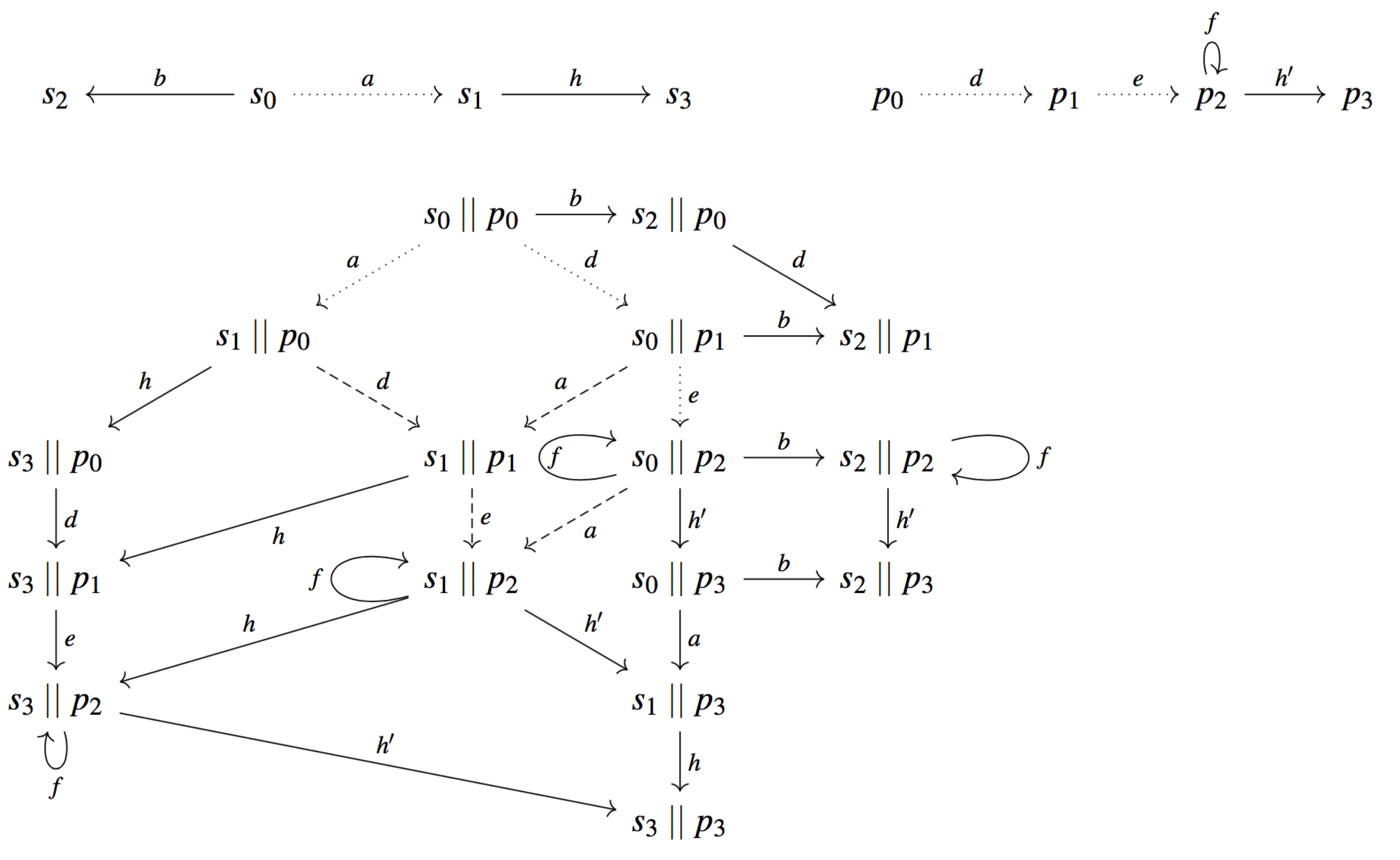}
\end{tabular} 
}
\caption{(De-)composing causality\label{fig:ex-dec}} 
\end{figure}
\end{example}

\subsection{(De-)Composing Conjunction}\label{sec:dec-conj}
In what follows we show that reasoning on causality with respect to conjunctions of HML formulae $\phi \land \psi$ can be performed in a compositional fashion.

Lemma~\ref{lm:decomposing-conjunction-if} states that causalities in two non-communicating LTS's are reflected within their interleaving as well.

\begin{lemma}[]\label{lm:decomposing-conjunction-if}
Consider LTS's $T = (\prc,s_0,\act,\rightarrow)$ and $T' = (\prc',s_0',B,\rightarrow')$ such that $A \cap B = \emptyset$.
Assume two HML formulae $\phi$ and $\psi$ over $A$ and $B$, respectively.
Whenever $\phi$ and $\psi$ are not immediate effects, the following holds. If
\begin{enumerate}
\item[ ] $\pi = (s_k , l_k, \decS_k), \ldots,  (s_m , l_m, \decS_m), s_{m+1} \in \mathit{Causes}(\phi, T )$ and\\
$\pi' = (s'_p , l'_p, \decS'_p), \ldots,  (s'_q , l'_q, \decS'_q), s'_{q+1} \in \mathit{Causes}(\psi, T')$ then\\
$\mu = (s_0 \Par s'_0, l_0, \overline{\decS}_0), \ldots,  (s_n \Par s'_n, l_n, \overline{\decS}_n), s_{n+1} \Par s'_{n+1} \in \mathit{Causes}(\phi \land \psi, T \Par T')$
\end{enumerate}
for all $\mu$ such that $s_0 \Par s'_0 \xrightarrow{l_0} \ldots s_n \Par s'_n \xrightarrow{l_n} s_{n+1} \Par s'_{n+1}$ is an execution sequence in $s_k \xrightarrow{l_k} \ldots s_m \xrightarrow{l_m} s_{m+1} \Par s'_p \xrightarrow{l'_p} \ldots s'_q \xrightarrow{l'_q} s'_{q+1}$, and
$s_0 \Par s'_0 = s_k \Par s'_p$, $s_n \Par s'_n = s_m \Par s'_q$, $s_{n+1} \Par s'_{n+1} = s_{m+1} \Par s'_{q+1}$,
$l_0 \ldots l_n \downarrow A = l_k \ldots l_m$ and $l_0 \ldots l_n \downarrow B = l'_p \ldots l'_q$.
\end{lemma}
\begin{proof}[Proof Sketch]
The statement is a consequence of two intermediate results.

First we show that whenever $\pi$ and $\pi'$ satisfy conditions AC1--AC2(c) in Definition~\ref{def:causality2}, one can build $\mu$ as above, such that $\mu$ satisfies AC1--AC2(c) as well.
Showing that $\mu$ satisfies AC1 and AC2 is immediate, by the assumption that both $\pi$ and $\pi'$ satisfy AC1--AC2(c) and the fact that $\phi$ and $\psi$ are not immediate effects.
Proving that AC2(b) and AC2(c) hold for $\mu$ strongly relies on the lists $\overline{\decS}_i$ in $\mu$.
The construction of $\overline{\decS}_i$'s is as follows.
\begin{enumerate}
\item We start with $\overline{\decS}_i$'s set to the empty list $[\,]$.
\item Then, note that all causal traces $\chi$ corresponding to $\pi$ are causal for $\neg \phi \lor \neg \psi$ as well. Hence, we consider sequences $\overline{\chi}$ from the interleaving of such $\chi$ with $\chi' \in B^*$
and make the corresponding additions to all $\overline{\decS}_i$'s, such that $\overline{\chi}$ is captured within $traces((l_0, \overline{\decS}_0) \ldots (l_n, \overline{\decS}_n))$ as well.
Symmetrically, repeat the procedure for all causal traces corresponding to $\pi'$.

Intuitively, this step works also as a "cleaning" step preparing $\mu$ to satisfy AC2(b) w.r.t. $\phi \land \psi$.
\end{enumerate}
At this point AC2(b) and AC2(c) hold for $\mu$, by the construction of lists $\overline{\decS}_i$ above.

Proving minimality of $\mu$ follows by reductio ad absurdum. The intuition is as follows. Whenever there exists $\mu' \in sub(\mu)$ such that $\mu'$ satisfies AC1--AC2(c), one can build $\widetilde{\pi} \in sub(\pi)$ and $\widetilde{\pi'} \in sub(\pi')$ such that $\widetilde{\pi}$ and $\widetilde{\pi'}$ satisfy AC1--AC2(c). This contradicts the hypothesis $\pi \in  \mathit{Causes}(\phi, T)$ and $\pi' \in  \mathit{Causes}(\psi, T')$.

\end{proof}


Lemma~\ref{lm:decomposing-conjunction-only-if} states that causality with respect to an HML formula $\phi \land \psi$ in the context of interleaved, non-communicating LTS's, determines causality with respect to $\phi$ and $\psi$ in the corresponding interleaved components.

\begin{lemma}[]\label{lm:decomposing-conjunction-only-if}
Consider LTS's $T = (\prc,s_0,\act,\rightarrow)$ and $T' = (\prc',s_0',B,\rightarrow')$ such that $A \cap B = \emptyset$.
Assume two HML formulae $\phi$ and $\psi$ over $A$ and $B$, respectively.
Whenever $\phi$ and $\psi$ are not immediate effects, the following holds.

If $\mu = (s_0 \Par s'_0, l_0, \overline{\decS}_0), \ldots,  (s_n \Par s'_n, l_n, \overline{\decS}_n), s_{n+1} \Par s'_{n+1} \in \mathit{Causes}(\phi \land \psi, T \Par T')$, 
then there exist
\begin{enumerate}
\item[ ] $\pi = (s_k , l_k, \decS_k), \ldots,  (s_m , l_m, \decS_m), s_{m+1} \in \mathit{Causes}(\phi, T )$ and\\
$\pi' = (s'_p , l'_p, \decS'_p), \ldots,  (s'_q , l'_q, \decS'_q), s'_{q+1} \in \mathit{Causes}(\psi, T')$
\end{enumerate}
where $s_k \xrightarrow{l_k} \ldots s_m \xrightarrow{l_m} s_{m+1} \Par s'_p \xrightarrow{l'_p} \ldots s'_q \xrightarrow{l'_q} s'_{q+1}$ includes the execution sequence
$s_0 \Par s'_0 \xrightarrow{l_0} \ldots s_n \Par s'_n \xrightarrow{l_n} s_{n+1} \Par s'_{n+1}$, and
$s_k \Par s'_p = s_0 \Par s'_0$, $s_m \Par s'_q = s_n \Par s'_n$, $s_{m+1} \Par s'_{q+1} = s_{n+1} \Par s'_{n+1}$,
$l_k \ldots l_m = l_0 \ldots l_n \downarrow A$ and $l'_p \ldots l'_q = l_0 \ldots l_n \downarrow B$.
\end{lemma}
\begin{proof}[Proof Sketch]
First, we show that one can build $\pi$ and $\pi'$ as above, such that $\pi$ and $\pi'$ satisfy conditions AC1--AC2(c) in Definition~\ref{def:causality2}, given the hypothesis that $\mu$ satisfies AC1--AC2(c) as well.
The reasoning for proving this intermediate result strongly relies on the shape of the lists ${\decS}_i$ and ${\decS}'_j$ corresponding to $\pi$ and $\pi'$, respectively.
We construct the aforementioned lists in three steps.
\begin{enumerate}
\item We start with empty lists ${\decS}_i$ and ${\decS}'_j$.
\item Then, we "encode" causal sequences $\chi \in traces((l_0, \overline{\decS}_0)\ldots(l_n, \overline{\decS}_n)) \setminus \{l_0 \ldots l_n\}$ satisfying AC2(c) by definition, into $traces((l_k, {\decS}_k)\ldots(l_m, {\decS}_m))$ and, respectively, $traces((l'_p, {\decS}'_p)\ldots(l'_q, {\decS}'_q))$ as follows.
Whenever $\chi$ always leads to states satisfying $\neg \phi$, make the corresponding additions to $\decS_i$ such that the projection of $\chi$ on $A$ is stored within $traces((l_k, {\decS}_k)\ldots(l_m, {\decS}_m))$.
Symmetrically, repeat the procedure for causal sequences $\chi$ that always lead to states satisfying $\neg \psi$.

\item Eventually, we "prepare" $\pi$ for satisfying AC2(b). 
We identify all sequences $\chi \in A^* \setminus traces((l_k, \overline{\decS}_k)$
$\ldots(l_m, \overline{\decS}_m))$ that always lead to $s \vDash \neg \phi$. 
For each such $\chi$ we make the necessary insertions into the lists $\overline{\decS}_i$, so that $\chi$ is stored as a causal trace of $\pi$. We repeat the "preparation" process for $\pi'$ as well.
\end{enumerate}
Then, we show that $\pi$ and $\pi'$ satisfy AC1--AC2(c) by reductio ad absurdum.
Showing that $\pi$ has to satisfy $AC3$ follows by proof by contradiction as well. Intuitively, we show that whenever there exists $\widetilde{\pi} \in sub(\pi)$ satisfying AC1--AC2(c), one can construct $\widetilde{\mu} \in sub(\mu)$ such that $\widetilde{\mu}$ satisfies AC1--AC2(c) as well. This contradicts the hypothesis $\mu \in \mathit{Causes}(\phi \land \psi, T \Par T')$.
Similar reasoning for proving that $\pi'$ has to satisfy $AC3$.

\end{proof}

Theorem~\ref{thm:decomposing-conjunction-iso} is the main result of this section.
Intuitively, it states that reasoning on causality with respect to an effect $\phi \land \psi$ in the context of non-communicating, interleaved LTS's is equivalent to reasoning on causality for $\phi$ and $\psi$ in the context of the corresponding interleaved components.

\begin{theorem}[(De-)composing Conjunction]\label{thm:decomposing-conjunction-iso}
Consider $T = (\prc,s_0,\act,\rightarrow)$ and $T' = (\prc',s_0',B,\rightarrow')$ such that $A \cap B = \emptyset$.
Assume two HML formulae $\phi$ and $\psi$ over $A$ and $B$, respectively.
Whenever $\phi$ and $\psi$ are not immediate effects, the following holds: 
\begin{equation}
T \Par T' \downarrow (\phi \land \psi) \,\,=\,\, 
(T \downarrow \phi) \Par (T' \downarrow \psi).
\end{equation}
\end{theorem}
\begin{proof}
The result is immediate by Lemma~\ref{lm:decomposing-conjunction-if} and Lemma~\ref{lm:decomposing-conjunction-only-if}.
\end{proof}

For an example, we refer again to the LTS's in Figure~\ref{fig:ex-dec}.
The causal projection $T \Par T' \downarrow (\phi \land \psi)$ is defined by the dashed/dotted transitions
$s_0 \Par p_0 \dotarrow{d} s_0 \Par p_1 \xdasharrow{~~$a$~~~} s_1 \Par p_1  \xdasharrow{~~$e$~~~} s_1 \Par p_2$, $s_0 \Par p_0 \dotarrow{d} s_0 \Par p_1 \dotarrow{e} s_0 \Par p_2  \xdasharrow{~~$a$~~~} s_1 \Par p_2$ and $s_0 \Par p_0 \dotarrow{a} s_1 \Par p_0 \xdasharrow{~~$d$~~~} s_1 \Par p_1  \xdasharrow{~~$e$~~~} s_1 \Par p_2$. This is precisely the interleaving of the causal projections $T \downarrow \phi$ and $T'\downarrow \psi$.


\begin{remark}\label{rm:OC}
As pointed out in Section~\ref{sec:def-cause}, the proposed notion of causality does not check whether the order in which certain actions are executed is causal with respect to the violation of a safety property, or not.
Nevertheless, as already mentioned, for non-interleaved systems such orderings are implicitly captured by sequences $l_0 \ldots l_n$ determined by causal computations as in Definition~\ref{def:causality2}. 
Additionally, in the context of interleaved systems, the ordering information can be irrelevant.
For formulae defined over disjoint alphabets, based on the compositionality results in Theorem~\ref{thm:decomposing-disjunction-iso} and Theorem~\ref{thm:decomposing-conjunction-iso}, causal reasoning is "pushed" at the level of the interleaved components, hence the order in which these components execute the interleaving does not matter.
\end{remark}

\section{Conclusions and Future Work}\label{sec:conclude}

In this paper we introduce a notion of causality for LTS's and violation of safety properties expressed in terms of HML formulae.
The proposed notion of causality inherits the characteristics of "actual causation" proposed in~\cite{halpern2005causes,DBLP:conf/vmcai/Leitner-FischerL13} and, in addition, is compositional with respect to the interleaving of 
the considered type of non-communicating LTS's.

A natural extension is handling causality in the context of communicating LTS's in the style of CCS~\cite{DBLP:books/sp/Milner80}, for instance.
The challenge would be to establish (de-)compositionality results whenever the interleaved systems display internal, non-observable behaviour.
The current approach  relies on the fact that the HML formulae are defined over "observable", disjoint alphabets. 
However, the general modal decomposition theorems such as those proposed in \cite{Larsen91,Aceto12} do provide support for arbitrary formulae and silent actions. This provides an interesting ground to extend our approach to communicating processes.

Of equal importance is extending our framework to handle causality for liveness properties as well. This can be achieved via HML with recursion, which is again treated in modal decomposition approaches \cite{Aceto12}.

We would also like to investigate the benefits of casting causality within a process algebraic setting. Observe that, for instance, causal projections can be naturally expressed as CCS process terms derived from CCS terms for components or their underlying LTS's.
Hence, we would like to study whether a process algebraic handling of causality provide more insight on its properties and whether causality as described in this paper can be axiomatized. 

Last, but not least, we would like to investigate to what extent our definition of causality is related to the actual causality in~\cite{DBLP:conf/vmcai/Leitner-FischerL13,DBLP:conf/spin/BeerHKLL15}.
As already discussed in the current paper, the two notions share similar characteristics, including causal non-occurrence of events and the ordering condition (that is implicit in our approach).
Once such a relationship is identified, one could exploit the compositionality results to improve fault localisation in automated tools for causality checking~\cite{DBLP:conf/vmcai/Leitner-FischerL13,DBLP:conf/spin/BeerHKLL15}.
\vspace{-15pt}
\paragraph{Acknowledgements}{
We thank the anonymous reviewers of CREST 2016 for their constructive comments and references to the literature.
The work of Georgiana Caltais was partially supported by an Independent Research Start-up Grant founded by Zukunftskolleg at Konstanz University.
The work of
Mohammad Reza Mousavi has been partially supported by the Swedish Research Council (Vetenskapsr{\aa}det) award number: 621-2014-5057 (Effective
Model-Based Testing of Concurrent Systems) and the Swedish Knowledge
Foundation (Stiftelsen f{\"o}r Kunskaps- och Kompetensutveckling) in
the context of the AUTO-CAAS H{\"o}G project (number: 20140312).}

\vspace{-10pt}

\end{document}